\documentclass[letter,11pt]{article}
\usepackage{fullpage}
\usepackage{microtype}
\usepackage{mathtools,amsmath} \usepackage{amssymb,amsthm}
\usepackage{graphicx}
\usepackage{hyperref,color}
\usepackage{lineno}
\usepackage{romannum}
 \usepackage[margin=1in]{geometry}

\newbox\ProofSym
\setbox\ProofSym=\hbox{%
	\unitlength=0.18ex%
	\begin{picture}(10,10)
	\put(0,0){\framebox(9,9){}}
	\put(0,3){\framebox(6,6){}}
	\end{picture}}

\graphicspath{{figures/}}

\newtheorem{theorem}{Theorem} 
\newtheorem{lemma}[theorem]{Lemma}

\newtheorem{corollary}[theorem]{Corollary}

\begin{document}
\pagenumbering{arabic}
\title{Algorithms for Computing Maximum Cliques in Hyperbolic Random Graphs{\footnote{This work was supported by the National Research Foundation of Korea(NRF) grant funded by the Korea government(MSIT) (No.RS-2023-00209069)}}}

%
%

\author{
Eunjin Oh\footnote{Pohang University of Science and Technology,	Korea.
		Email: {\tt{eunjin.oh@postech.ac.kr}}}
\and
Seunghyeok Oh\footnote{Pohang University of Science and Technology,	Korea.
		Email: {\tt{seunghyeokoh@postech.ac.kr}}}}

\maketitle

\begin{abstract}

In this paper, we study the maximum clique problem on hyperbolic random graphs. 
A hyperbolic random graph is a mathematical model for analyzing scale-free networks since it effectively explains the power-law degree distribution of scale-free networks. 
We propose a simple algorithm for finding a maximum clique in hyperbolic random graph. 
We first analyze the running time of our algorithm theoretically. 
We can compute a maximum clique on a hyperbolic random graph $G$ in $O(m + n^{4.5(1-\alpha)})$ expected time 
if a geometric representation is given 
or in $O(m + n^{6(1-\alpha)})$ expected time if a geometric representation is not given, where $n$ and $m$ denote the numbers of vertices and edges of $G$, respectively, and $\alpha$ denotes a parameter controlling the power-law exponent of the degree distribution of $G$.
Also, we implemented and evaluated our algorithm empirically.
Our algorithm outperforms the previous algorithm [BFK18] practically and theoretically. 
Beyond the hyperbolic random graphs, we have experiment on real-world networks. 
For most of instances, we get large cliques close to the optimum solutions efficiently.

\end{abstract} 
\newpage

\section{Introduction}
Designing algorithms for analyzing large real-world networks
such as social networks, World Wide Web, or
biological networks 
is a fundamental problem in computer science that has attracted considerable attention in the last decades. 
To deal with this problem 
from the theoretical point of view, 
it is required to define a mathematical model for real-world networks.
For this purpose, several models have been proposed. Those models are required to 
replicate the salient features of real-world networks. 
One of the most salient features of real-world networks is \emph{scale-free}. 
In general, a graph is considered as a scale-free network if its diameter is small, one connected component has large size, it has subgraphs with large edge density, and most importantly, \emph{its degree distribution follows a power law}.
Here, for an integer $k\leq n$, let $P(k)$ be the fraction of nodes having degree exactly $k$.
If $P(k) \sim k^{-\beta}$, we say that the degree distribution of the graph follows a power law.
In this case, $\beta$ is called the \emph{power-law exponent}. 

One promising model for scale-free real-world networks is
the hyperbolic random graph model. 
A hyperbolic random graph is constructed from two parameters.
First, points in the hyperbolic plane are chosen from a certain distribution depending on the parameters.
Then we consider the points as the vertices of the constructed hyperbolic random graph.
For two vertices whose distance is at most a certain threshold, we add an edge between them. For illustration, see Figure~\ref{fig:hyperbolic}. 
It is known that the degree distribution of a hyperbolic random graph follows a  power-law~\cite{gugelmann2012random}. 
Moreover, its diameter is small with high probability~\cite{friedrich2018diameter,kiwi2014bound}, and it has a giant connected component~\cite{bode2015largest,doi:10.1137/18M121201X}. 
Including these results, the structural properties of hyperbolic random graphs have been
studied extensively~\cite{candellero2014clustering}. 
However, only a few algorithmic results are known. 
In other words, the previous work focuses on \emph{why} 
we can use hyperbolic random graphs as a promising model,
but only a few work focuses on \emph{how} to use this model for solving real-world problems. 
We focus on the latter type of problems. 

In this paper, we focus on the maximum clique problem for hyperbolic random graphs from  theoretical and practical point of view. 
The maximum clique problem asks for a maximum-cardinality set of pairwise adjacent vertices. 
For general graphs, this problem is NP-hard. Moreover, it is W[1]-hard when it is parameterized by the solution size, and it is APX-hard even for cubic graphs~\cite{alimonti2000some}. Therefore, the theoretical study on the clique problem focuses on special classes of graphs. 
In fact, this problem can be solved in polynomial time for special classes of graphs such as
planar graphs, unit disk graphs and hyperbolic random graphs~\cite{blasius2018cliques,CLARK1990165,espenant2023finding}.
More specifically, the algorithm by Bl\"asius et al.~\cite{blasius2018cliques} for  computing a maximum clique of a hyperbolic random graph takes $O(mn^{2.5})$ \emph{worst-case} time. The randomness in the choice of vertices is not considered in the analysis of the algorithm.
One natural question here is to design an algorithm for this problem with improved \emph{expected} time.

To analyze our algorithm, we use the \emph{average-case analysis}.
A traditional modeling choice in the analysis of algorithms is \emph{worst-case analysis}, where the performance of an algorithm is measured by the worst performance over all possible inputs.
Although it is a useful framework in the analysis of algorithms,  
it does not take into account the distribution of inputs that an algorithm is likely to encounter in practice. 
It is possible that an algorithm performs well on most inputs, but poorly on a small number of inputs that are rarely encountered in practice. In this case, the worst-case analysis can mislead the analysis of algorithms. The field of ''beyond worst-case analysis''
studies ways for overcoming these limitations~\cite{roughgarden2021beyond}. 
One simple technique studied in this field is the average-case analysis. 
As the hyperbolic random graph model intrinsically defines an input distribution, 
the average-case analysis is a natural model for analyzing algorithms for hyperbolic random graphs. 

\subparagraph{Previous Work.} 
While the structural properties of hyperbolic random graphs has been studied extensively, only a few algorithmic problems have been studied.
The most extensively studied algorithmic problem on hyperbolic random graphs
is the generation problem:
Given parameters, the goal is to generate a hyperbolic random graph efficiently. 
The best-known algorithms run in expected linear time~\cite{blasius2022efficiently} and in worst-case subquadratic time~\cite{von2015generating}. 
Also, Bl\"asius et al.~\cite{blasius2018efficient} studied the problem of 
embedding a scale-free network into the hyperbolic plane and presented heuristics for this problem. 

Very recently, classical algorithmic problems such as shortest path problems, the maximum clique problem and the independent set problem have been studied.
These problems can  be solved significantly faster in hyperbolic random graphs. 
More specifically, given a hyperbolic random graph, the shortest path between any two vertices can be computed in sublinear expected time~\cite{blasius2022efficient, blasius2021efficiently}.
A hyperbolic random graph admits a sublinear-sized balanced separator with high probability~\cite{blasius2016hyperbolic}. 
As applications, Bl\"asius et al.~\cite{blasius2016hyperbolic} showed that 
the independent set problem admits a PTAS for hyperbolic random graphs,
and the maximum matching problem admits a subquadratic-time algorithm. 
Also, the clique problem can be solved in polynomial time for hyperbolic random graphs 
in the worst case~\cite{blasius2018cliques}. 

The clique problem has been studied extensively because it has numerous applications in various field
such as community search in social networks, team formation in expert networks, gene expression and motif discovery in bioinformatics and anomaly detection in complex networks~\cite{lu2017finding}.
From a theoretical perspective, the best-known exact algorithm runs in 
$2^{0.276n}$ time in~\cite{robson1986algorithms}. 
However, it is not sufficiently fast for massive real-world networks,
leading to the proposal of 
lots of exact algorithms and heuristics for this problem on real-world networks~\cite{abello2002massive,lu2017finding,pattabiraman2013fast,rossi2014fast}. 
While these algorithms and heuristics work efficiently in practice, there is no theoretical guarantee of their efficiency. 

\subparagraph{Our results.}
We present algorithms for computing a maximum clique in a hyperbolic random graph and analyze their performances theoretically and empirically.

Given a hyperbolic random graph with parameters $\alpha$ and $C$ together with its geometric representation, we can compute 
a maximum clique in $O(m+n^{4.5(1-\alpha)})$ expected time, where $n$ and $m$ denotes the numbers of vertices and edges of the given graph. Here, we have $1/2<\alpha<1$, and the O-notation hides a constant depending on $\alpha$ and $C$. 
With high probability, our algorithm outperforms the previously best-known algorithm by Bl\"asius et al.~\cite{blasius2018cliques} running in $O(mn^{2.5})$ time. 
In the case that a geometric representation is not given,  
our algorithm runs in $O(m+n^{6(1-\alpha)})$ expected time.
This is the first algorithms for the maximum clique problem
on hyperbolic random graphs not using geometric representations. 

Also, we implemented our algorithms
and analyzed it empirically. We run our algorithms on both
synthetic data (hyperbolic random graphs) and real-world data. 
For hyperbolic random graphs, since it is proved that our algorithm computes a maximum clique correctly, we focus on the efficiency of the algorithms. 
We observed that our algorithms perform efficiently; it takes about 100ms for $n=10^6$. For real-world networks,
our algorithm gives a lower bound on the optimal solution. 
We observed that our algorithm performs well especially for collaboration networks and web networks. These are typical scale-free real-world networks.

\section{Preliminaries}
Let $\mathbb{H}^2$ be the hyperbolic plane with curvature $-1$. We can handle hyperbolic planes with arbitrary (negative) curvatures by rescaling other model parameters which will be defined later. Thus it suffices to deal with the hyperbolic plane with curvature $-1$. 
Since the hyperbolic plane is isotropic, we choose an arbitrary point $o$ and consider it as the origin of $\mathbb{H}^2$. Also, we fix a half-line $\ell_o$ from $o$ going towards an arbitrary point, say $w$, as the \emph{axis}. Then we can represent a point $v$ of $\mathbb{H}^2$ as $(r_v, \phi_v)$ where $r_v$ is the hyperbolic distance between $v$ and $o$, and $\phi_v$ is the angle from $\ell_o$ to the half-line from $o$ going towards $v$.
We call $r_v$ and $\phi_v$ the \emph{radial} and
\emph{angular coordinates} of $v$. 

For any two points $x$ and $y$ in $\mathbb{H}^2$,
we use $d(x,y)$ to denote the distance in $\mathbb{H}^2$ between $x$ and $y$. 
Then we have the following. 
\begin{align*}
    d(u, v) &= \cosh^{-1} (\cosh (r_u) \cosh (r_v) - \sinh(r_u) \sinh(r_v)\cos(\Delta \phi_{u,v})\\
    &\leq\cosh^{-1} (\cosh (r_u) \cosh (r_v)+ \sinh(r_u) \sinh(r_v)), 
\end{align*}
where $\Delta \phi_{u,v}$ denotes the small relative angle between $v$ and $u$~\cite{anderson2006hyperbolic}. 
For a point $x\in\mathbb{H}^2$ and a value $r\geq 0$, 
we use $B_x(r)$ to denote 
the disk centered at $x$ with radius $r$.
That is, 
$B_x(r)=\{v\in\mathbb{H}^2 \mid d(v,x)\leq r\}$. 

\begin{figure}
    \centering
    \includegraphics[width=0.8\linewidth]{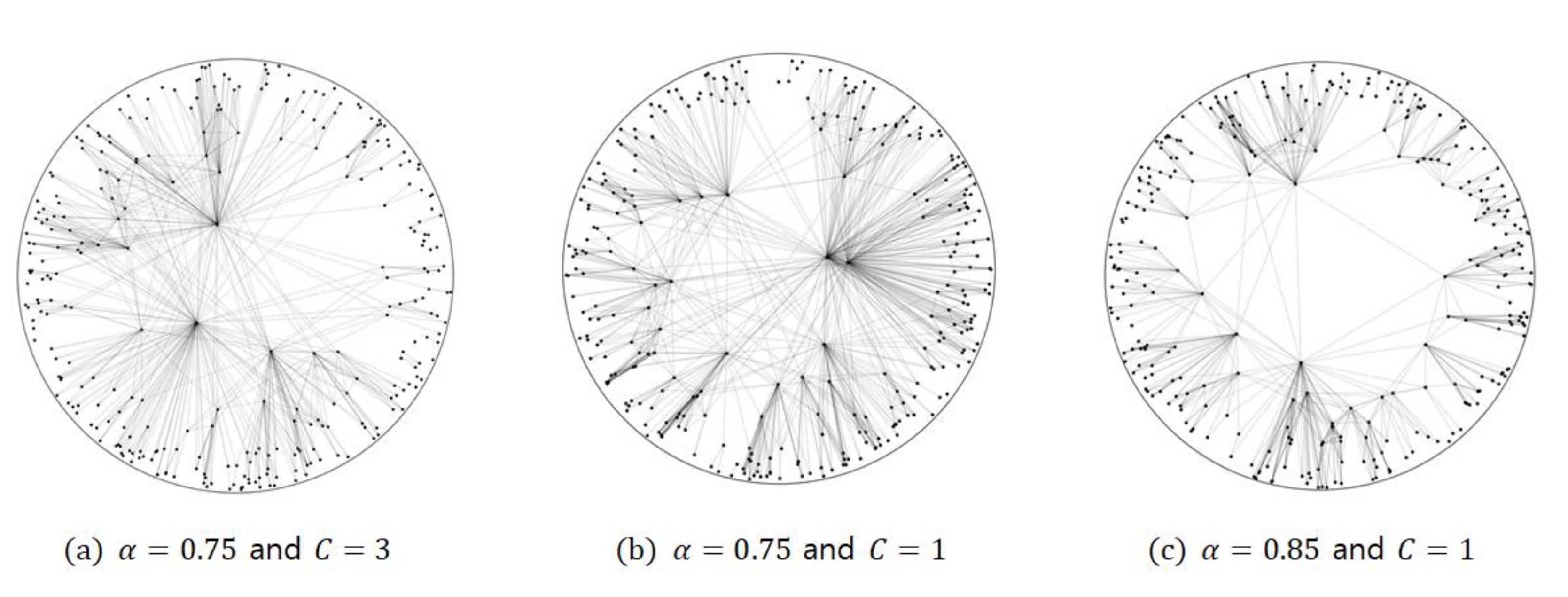}
    \caption{
    HRGs with different parameters. 
    As $C$ gets larger, the average degree gets larger (See (a--b)).
    As $\alpha$ gets larger, the points gets closer to the boundary of $D$ (See (b--c)).
    }
    \label{fig:hyperbolic}
    \vspace{-0.5em}
\end{figure}

\subsection{Hyperbolic Random Graphs}
A \emph{hyperbolic unit disk graph (HRG)}  is a graph whose vertices are placed on $\mathbb{H}^2$, and two vertices are connected by an edge if the distance between them on $\mathbb{H}^2$ is at most some threshold. This threshold is called the \emph{radius threshold} of the graph. This is the same as the unit disk graph except that the hyperbolic unit disk graph is defined on $\mathbb{H}^2$ while the unit disk graph is defined on the Euclidean plane. 

In this paper, we focus on the hyperbolic \emph{random} graph model
introduced by Papadopoulos et al.~\cite{papadopoulos2010greedy}.
It is a family $\{\mathcal G_{n, \alpha, C}\}$ of distributions, indexed by the number $n$ of vertices, a parameter $C$ adjusting the average degree of a graph, and a parameter $\alpha$ determining the power-law exponent. 
A sample from $\{\mathcal G_{n, \alpha, C}\}$ is a 
hyperbolic unit disk graph on $n$ points (vertices) chosen independently as follows. 
Let $D$ be the disk centered at $o$ with radius $R=2\ln n+C$. 
To pick a point $v$ in $D$, we first sample its radius $r_v$, and then sample its angular coordinate $\phi_v$.
The probability density $\rho(r)$ for the radial coordinate $r_v$ is defined as 
\begin{align}\label{eqn:rho}
    \rho(r) &= \frac{\alpha\sinh(\alpha r)}{\cosh(\alpha R)-1}.
\end{align}
Then the angular coordinate is sampled uniformly from $[0,2\pi)$. 
In this way, we can sample one vertex with respect to parameters $\alpha$ and $C$,
and by choosing $n$ vertices independently and by computing the hyperbolic 
unit disk graph with radius threshold $R$ on the $n$ vertices, 
we can obtain a sample from the distribution  $\{\mathcal G_{n, \alpha, C}\}$. 

An intuition behind the definition of $\rho(r)$ is as follows.
To choose a point $v$ uniformly at random in $D$, we first choose its angular coordinate uniformly at random in $[0,2\pi)$ as we did for $\mathcal{G}_{n,\alpha,C}$, and choose its radial  coordinate according to the distribution with density function $\rho(r)=\frac{\sinh(r)}{\cosh(\alpha R)-1}$. But in this case, the power-law exponent of a graph is fixed.
To add the flexibility to the model, the authors of~\cite{papadopoulos2010greedy} introduced a parameter $\alpha$ and defined 
the density function as in~\eqref{eqn:rho}. Here, 
For $\alpha<1$, this favors points closer to the center of $D$, while for $\alpha>1$, this favors points closer to the boundary of $D$. For $\alpha=1$, this corresponds to the uniform distribution~\cite{gugelmann2012random}.
For illustration, see Figure~\ref{fig:hyperbolic}. 

\subparagraph{Properties of HRGs.}
Let $\mu(S)$ be the probability measure of a set $S\subseteq D$, that is,
\begin{align*}
    \mu(S) &= \int_{x\in S}  \frac{\alpha\sinh(\alpha r_x)}{2\pi \cosh(\alpha R)-1} dx.
\end{align*}
For a vertex $v$ of a hyperbolic random graph $G$ with $n$ vertices,
the expected degree of $v$ in $G$ is $n\cdot \mu(B_v(R)\cap D)$ by construction. 
Moreover, notice that $\mu(B_v(R)\cap D)=\mu(B_{v'}(R)\cap D)$ for any two vertices $v, v'$ with $r_v=r_{v'}$. Thus to make the description easier, we let $B_r(r')$ denote $B_{(r,0)}(r')$ if it is clear from the context. Note that $D=B_0(R)$. 
For theoretical analyses of our algorithm, we analyze the probability measures for different sets using Lemma~\ref{lem:expected-degree}.


\begin{lemma}[{\cite{blasius2018cliques}}]\label{lem:expected-degree}
For any $0 \leq r, r' \leq R$, we have
\begin{align*}
    \mu (B_0(r)) &= e^{-\alpha(R-r)} (1 - \mathcal{E}_r)&\\ 
    \mu (B_r(R) \cap B_0(R)) &= \frac{2\alpha e^{-r/2}}{\pi(\alpha - 1/2)} (1 - \mathcal{E}_r')&\\
    \mu (B_{r}(R) \cap B_0(r)) &= 
    \begin{cases}
    \frac{2\alpha}{\pi (\alpha - 1/2)} e^{-R(\alpha-1/2) + r(\alpha-1)} (1 - o(1))& \text{if } r > R/2,\\
    \mu(B_0(r)) & \text{if } r \leq 
    R/2,
    \end{cases}
\end{align*}
where $\mathcal{E}_r$ and $\mathcal{E}_r'$ are error terms which can be negligible (that is, $O(1/n)$) if $r=\Omega(R)$. 
\end{lemma}

Hyperbolic random graphs have all properties for being considered as scale-free networks mentioned above.
In particular, hyperbolic random graphs  with parameter $\alpha$ have power-law exponent $\beta$, 
where $\beta=2\alpha+1$ if $\alpha\geq 1/2$, and $\beta=2$ otherwise. 
Most real-world networks have a power-law exponent larger than two. Thus 
we assume that $\alpha>1/2$ in the paper. 

\subsection{Algorithms for the Maximum Clique Problem}
In this section, we review the algorithm for this problem on hyperbolic random graphs described  in~\cite{blasius2018cliques}, which is an extension of~\cite{CLARK1990165}. 
This algorithms requires geometric representations of hyperbolic random  graphs.
Let $G=(V,E)$ be a hyperbolic random  graph.  

For $\alpha \geq 1$, they 
showed that a hyperbolic random graph has $O(n)$ maximal cliques with high probability. Therefore, a maximum clique can be computed in linear time with high probability by just enumerating all the maximal cliques. 

For $\frac{1}{2} < \alpha < 1$, they showed that 
the algorithm in~\cite{CLARK1990165} can be extended to hyperbolic random graphs. 
Assume first that, for a maximum clique $K$, we have two vertices $u$ and $v$ of $K$ with maximum $r=d(u,v)$. 
Then all vertices in $K$ are contained in the region $R_{uv} = B_u(r) \cap B_v(r)$.  
Then we can compute $K$ by considering the vertices in $R_{uv}$ as follows. We partition $R_{uv}$ into $R_{uv}^1$ and $R_{uv}^2$ with respect to the line through $u$ and $v$. 
They showed that the diameter of $R_{uv}^1$ (and $R_{uv}^2$) is at most one, and thus 
$V\cap R_{uv}^1$ (and $V\cap R_{uv}^2)$) forms a clique.
Therefore, the subgraph $G_{uv}$ of $G$ induced by $V\cap R_{uv}$ is the complement of a bipartite graph  
with bipartition $(V\cap R_{uv}^1, V\cap R_{uv}^2)$. 
Moreover, $K$ is an independent set of the complement of $G_{uv}$. 
Therefore, it suffices to compute an independent set of the complement of $G_{uv}$, and we can do this in 
in $O(n^{2.5})$ time using the Hopcroft-Karp algorithm.
However, we are not given the edge $uv$ in advance. Thus we apply this procedure for every edge $uv$ of $G$, and then
take the largest clique as a solution. This takes  $\mathcal{O}(mn^{2.5})$ time in total. 


\medskip 
Throughout this paper, we use $\mathbb P[A]$ to denote the probability that
an event $A$ occurs. 
For a random variable $X$, we use $\mathbb E[X]$ to denote the 
expected value of $X$. 
We use the following form of Chernoff bound.
\begin{theorem}[{Chernoff Bound~\cite{dubhashi2009concentration}}]
Let $X_1, \ldots, X_n$ be independent random variables with $X_i \in \{ 0,1 \}$ and let $X = \Sigma_{i=1}^n X_i$.
Then we have 
\begin{align*}
\mathbb{P}[X>t] &\leq 2^{-t} && \text{for any constant } t > 2e\cdot \mathbb{E}[X] \text{ and} \\
\mathbb{P}[X<(1-\epsilon)\cdot \mathbb{E}[X]] &\leq e^{-\epsilon^2/2 \cdot \mathbb{E}[X]} && \text{for any constant } \epsilon \in (0, 1).
\end{align*}
\end{theorem}

\section{Efficient Algorithm for the Maximum Clique Problem}\label{sec:nonrobust-algo}
In this section, we present an algorithm for the maximum clique problem on a hyperbolic random graph drawn from $\mathcal{R}_{n,\alpha,C}$ running in $\mathcal{O}\left(m + n^{4.5(1-\alpha)} \right)$ expected time. This algorithm correctly works for any hyperbolic unit disk graph, but its time bound is guaranteed only for hyperbolic random graphs. As we only deal with the case that $1/2<\alpha<1$, this algorithm is significantly faster than the algorithm in~\cite{blasius2018cliques}.

A main observation is the following. 
A clique of size $k$ consists of vertices of degree at least $k-1$. 
That is, to find a clique of size at least $k$, removing vertices with degree less than $k-1$ does not affect the solution. 
Thus once we have a lower bound, say $k$, on the size of a maximum clique, we can remove
all vertices of degree less than $k$.
Our strategy is to construct a sufficiently large clique (which is not necessarily maximum)
 as a preprocessing step 
so that we can remove a sufficiently large number of vertices of small degree. 
After applying a preprocessing step, we will see that 
the number of vertices we have decreases to $O(n^{1-\alpha})$ with high probability. 
Then we apply the algorithm in~\cite{blasius2018cliques} to the resulting graph.

\subsection{Computing a Sufficiently Large Clique Efficiently}\label{sec:initial}
In this section, we show how to compute a clique of size $\Omega(n^{1-\alpha})$ with probability $1-2^{-\Omega(n^{1-\alpha})}$. 
The algorithm is simple: Scan the vertices in the non-increasing order of their degrees, and maintain a clique $Q$, which is initially set as $\emptyset$. 
If the next vertex can be added to $Q$ to form a larger clique, then add it, otherwise exclude it. 
We can sort the vertices with respect to their degrees in $O(n+m)$ time using counting sort.
 Also, we can construct the clique $Q$ 
 in $O(n+m)$ time.
 In the following, we call $Q$ the \emph{initial clique}.

Now, we show that the size of the initial clique is $\Omega(n^{1-\alpha})$ with probability $1-2^{-\Omega(n^{1-\alpha})}$. 
First, we show that a sufficient large clique can be found by collecting all vertices in $B_0(R/2)$  with high probability in Lemma~\ref{lem:clique-size}. 

\begin{lemma}\label{lem:clique-size}
    For any constant $c'>0$, the vertices in $B_0(R/2-c')$ forms a clique of size $\Omega(n^{1-\alpha})$ with probability $1-2^{-\Omega(n^{1-\alpha})}$. 
\end{lemma}

\begin{proof}
    For any two vertices $u$ and $v$ in $B_0(R/2)$, 
    we have 
    \begin{align*}
        d(u, v) \leq\cosh^{-1} (\cosh^2 (R/2)) \leq \cosh^{-1}(\cosh R) = R.
    \end{align*}
    That is, $u$ and $v$ are connected by an edge.
    Therefore, the vertices in $B_0(R/2)$  form a clique. 

    Now we show that the number of vertices in $B_0(R/2-c')$ is at most $\Omega(n^{1-\alpha})$ with high probability.
    For this purpose, we define a random variable $X$ as the number of vertices in $B_0(R/2-c')$. Then by Lemma~\ref{lem:expected-degree}, we have 
    \begin{align*}
        \mathbb{E}[X]  = n \mu ( B_0 (R/2-c') ) = ne^{-\frac{1}{2}\alpha (R+c')} (1-o(1)).
    \end{align*}
    Applying $R = 2\ln n + C$, we have 
    \begin{align*}
        \mathbb{E}[X] &= n^{1-\alpha} e^{-\frac{1}{2}\alpha (C+c')} (1-o(1)) \geq cn^{1-\alpha} 
    \end{align*}
    for some constant $c$.
    By Chernoff bound, we have
    \begin{align*}
        \mathbb{P}[X < (1-\epsilon)cn^{1-\alpha}] \leq e^{-\epsilon^2/2 \cdot cn^{1-\alpha}} = 2^{-\Omega(n^{1-\alpha})}
    \end{align*}
    for any constant $\epsilon \in (0, 1)$.
\end{proof}

Thus, we can get desired clique by choosing $\Omega(n^{1-\alpha})$ vertices in $B_0(R/2)$ with high probability. 
However, as we scan the vertices in the decreasing 
 order of their degrees, 
 this does not immediately imply that the size of the initial clique is $\Omega(n^{1-\alpha})$. 
In the following lemma, we show that the initial clique has the claimed size by showing that
the $\Omega(n^{1-\alpha})$ vertices with highest degrees are 
contained in $B_0(R/2)$ with high probability.

\begin{lemma}\label{lem:initial}
    The initial clique has size  
    $\Omega(n^{1-\alpha})$ with probability $1-2^{-\Omega(n^{1-\alpha})}$.
\end{lemma}
\begin{proof}
    First, we show that no vertex lying outside of $B_0(R/2)$ has degree greater than $2ec_1\sqrt{n}$ with high probability for some constant $c_1$ which will be defined later.
    For this purpose, we analyze the expected degree of 
    a vertex $v$ whose radius coordinate $r$ is larger than $R/2$. 
    Let $X$ be the degree of $v$. Then we have 
    \begin{align*}
        \mathbb{E}[X] &= n\mu(B_{r}(R) \cap B_0(R))
        = n\frac{2\alpha e^{-r/2}}{\pi(\alpha-1/2)}(1-o(1)) 
        \leq n\frac{2\alpha e^{-R/4}}{\pi(\alpha-1/2)}(1-o(1)). 
    \end{align*}
    for $r=R/2$.
    By applying $R = 2\ln n + C$, we have 
    \begin{align*}
        \mathbb{E}[X] \leq c_1\sqrt{n} 
    \end{align*}
    for $c_1 = \frac{2\alpha e^{-C/4}}{\pi(\alpha-1/2)}$.
    By Chernoff bound, we have 
    \begin{align*}
        \mathbb{P}[X > 2ec_1\sqrt{n}] \leq 2^{-2ec_1\sqrt{n}}.
    \end{align*}
    
    Therefore, 
    the probability that all vertices lying outside of $B_0(R/2)$ have degree larger than $2ec_1\sqrt{n}$  is at most $2^{-2ec_1\sqrt{n}}\times n = 2^{-\Omega(\sqrt{n})}$. 
    In other words, there is no vertex of degree greater than $2e c_1 \sqrt{n}$ in $r \geq \frac{R}{2}$ with probability $1 - 2^{-\Omega(\sqrt{n})}$.

    Now we show that 
    no vertex  in $B_0(R/2-c_2)$ has degree smaller than $2ec_1\sqrt{n}$ with high probability
    for some constant $c_2$ which will be specified later.  
    For this purpose, we analyze the probability that
    a vertex $v$  in $B_0(R/2-c_2)$ has degree smaller than $2ec_1\sqrt{n}$. 
    Let $Y$ be the degree of $v$.  
    Then we have 
    \begin{align*}
        \mathbb{E}[Y] &= n\mu(B_{r}(R) \cap B_0(R)) 
        = n\frac{2\alpha e^{-r/2}}{\pi(\alpha-1/2)}(1-o(1))
        \geq n\frac{2\alpha e^{-R/4}e^{c_2/2}}{\pi(\alpha-1/2)}(1-o(1))\\
        &> \sqrt{n} \cdot 
        \frac{2\alpha e^{-C/4}e^{c_2/2}}{\pi(\alpha-1/2)}\cdot\frac{1}{2}\\
        &> e^{c_2/4} c_1\sqrt{n}.
    \end{align*}
     Here, $c_1$ is the constant we used for the previous case. By Chernoff bound, we have 
    \begin{align*}
        \mathbb{P}[Y < (1-\epsilon)e^{c_2/4}c_1\sqrt{n}] \leq 2^{-\epsilon^2/2 \cdot e^{c_2/4}c_1\sqrt{n}} 
    \end{align*}
    for any constant $\epsilon \in (0, 1)$.
    By choosing $c_2 = 4 + 4\ln \frac{2}{1-\epsilon}$, we get 
    \begin{align*}
        \mathbb{P}[Y < 2ec_1\sqrt{n}] \leq 2^{-\frac{e\epsilon^2}{1-\epsilon}c_1\sqrt{n}}.
    \end{align*}
    Therefore, 
    the probability that all vertices contained in $B_0(R/2-c_2)$ have degree smaller than $2ec_1\sqrt{n}$  is at most $2^{-\frac{e\epsilon^2}{1-\epsilon}c_1\sqrt{n}} \times n = 2^{-\Omega(\sqrt{n})}$. 
    In other words, there is no vertex of degree smaller than $2e c_1 \sqrt{n}$ in $r \geq \frac{R}{2}$ with probability $1 - 2^{-\Omega(\sqrt{n})}$.
    

    To construct the initial solution, we scan the vertices in the decreasing order of their degrees. Therefore,
    with probability $1 - 2^{-\Omega(\sqrt{n})}$,
    we consider all vertices in $B_0(R/2-c_2)$ before considering
    any vertex lying outside of $B_0(R/2)$. 
    Therefore, the initial clique contains all vertices in $B_0(R/2-c_2)$ with high probability. 
    By Lemma~\ref{lem:clique-size}, the initial clique has size at least $\Omega(n^{1-\alpha})$ with high probability.
\end{proof}

\subsection{Removing All Vertices of Small Degree}
\label{sec:degree-reduction}
In this section, we show how to remove a sufficiently large number of vertices, 
and show that the size of the remaining graph is $O(n^{1-\alpha})$ with high probability. 
This algorithm is also simple: 
given the initial clique of size $k$, 
we repeatedly delete all vertices of degree smaller than $k$. 
We call the resulting graph the \emph{kernel}. 
Then no vertex in the kernel has degree smaller than $k$ at the end of the process.
This process can be implemented in linear time as follows:
maintain the queue of vertices of degree smaller than $k$, and maintain the degree of each vertex. Then remove the vertices in the queue in order. Whenever a vertex $v$ is removed, update 
the degree of each neighbor $w$ of $v$ and insert $w$ to the queue if its degree gets smaller than $k$.

We show that the kernel has size  $O(n^{1-\alpha})$. 
Notice that we do not specify the order of vertices we consider during the deletion process. 
Fortunately, the kernel size remains the same regardless of the choice of deletion ordering. 

\begin{lemma}\label{lem:unique}
    In any order of deleting vertices, we can get a unique kernel.
\end{lemma}
\begin{proof}
    Let $G=(V,E)$ be a given graph. 
    Let $(a_1, a_2, \ldots, a_p )$
    and $(b_1, b_2, \ldots, b_q)$  be two different sequences of vertices
    deleted during the deletion process. 
    We first show that a \emph{set} $V_a = \{ a_1, a_2, \ldots, a_n \}$ 
    contains $V_b = \{ b_1, b_2, \ldots, b_m \}$. 
    If $b_1 \notin V_a$, the degree of $b_1$ in $G[V \setminus V_a]$ is less than $k$ 
    because removing vertices does not increase the degree of $b_1$, where
    $G[V']$ denotes the subgraph of $G$ induced by a vertex set $V'$. 
    That is, 
    if $b_1\in V\setminus V_a$, the deletion process would not terminate
    for $V_a$, and thus we conclude that $b_1 \in V_a$. 
    Inductively, assume that $\{b_1, b_2, \ldots b_i \} \subseteq V_a$.
    If $b_{i+1} \notin V_a$, the degree of $b_{i+1}$ in $G[V \setminus V_a]$ is at most the degree of $b_{i+1}$ in $G[V \setminus \{ b_1, b_2, \ldots, b_i \}]$,
    and thus 
    the degree of $b_{i+1}$ in $G[V \setminus V_a]$ is less than $k$.
    Then the deletion process would not terminate for $V_a$, and thus 
    $\{b_1, b_2, \ldots b_{i+1} \} \subseteq V_a$.
    Therefore, we have $V_b \subseteq V_a$. 
    Symmetrically, we can also show that $V_a \subseteq V_b$, and thus $V_a = V_b$.
    Therefore, the any order of deleting vertices gives the same remaining vertices.
\end{proof}

Because of the uniqueness of the kernel, 
for analysis, 
we may fix a specific deletion ordering and slightly modify the deletion process as follows.
Imagine that we scan the vertices in the decreasing order of their radial coordinates.
If the degree of a current vertex (in the remaining graph) is at least the size of the initial clique size, then we terminate the 
deletion process.
Otherwise, we delete the current vertex, and consider the next vertex.
By Lemma~\ref{lem:unique}, the number of remaining vertices is at least the size of the kernel. 
In the following lemma, we analyze the number of remaining vertices.

\begin{lemma}\label{lem:kernel-size}
    Given an initial solution of size $\Omega(n^{1-\alpha})$, then the size of the kernel is $O(n^{1-\alpha})$ with probability $1-2^{-\Omega(n^{1-\alpha})}$.
\end{lemma}

\begin{proof}
    Let $c_3$ be a sufficiently small constant such that the size of the initial solution is at least $c_3 n^{1-\alpha}$.
    Let $r'$ be the radius satisfying $n\mu (B_0(r') \cap B_{r'}(R)) = \frac{c_3}{2e} n^{1-\alpha}$.
    By choosing $c_3$ as a sufficient small constant, we may assume that $r'\geq R/2$. 
    In the following, we show that all the vertices with radial coordinates larger than $r'$ are removed with probability $1 - 2^{-\Omega(n^{1-\alpha})}$. 
    Observe that $\mu(B_0(r)\cap B_r(R))$ is decreasing for $r\geq R/2$. That is, 
    for $r > r'$,
    \begin{align*}
        \mu(B_0(r) \cap B_{r}(R)) &= \frac{2\alpha}{\pi(\alpha-1/2)} e^{-R(\alpha - 1/2) + r(\alpha -1)} \cdot (1-o(1))\\
        &\leq \frac{2\alpha}{\pi(\alpha-1/2)} e^{-R(\alpha - 1/2) + r'(\alpha -1)} \cdot (1- o(1))\\
        &\leq \mu(B_0(r') \cap B_{r'}(R)).
    \end{align*}
    For a vertex $v$, we define the \emph{inner degree} of $v$ as the number of its neighboring vertices whose radial coordinates are smaller than the radial coordinate of $v$.
    Let $X$ be the inner degree of $v$. Then the expected inner degree of a vertex
    with radial coordinate $r \geq r'$ is as follows.
    \begin{align*}
        \mathbb{E}[X] &\geq n \mu(B_0(r') \cap B_{r'}(R)) = \frac{c_3}{2e} n^{1-\alpha}.
    \end{align*}
    By Chernoff bound,
    \begin{align*}
        \mathbb{P}[X > c_3 n^{1-\alpha}] \leq 2^{-c_3 n^{1-\alpha}}.
    \end{align*}
    Therefore, for a vertex with radial coordinate larger than $r'$,
    the probability that its inner degree is larger than the size of the initial clique
    is at most $2^{-c_3 n^{1-\alpha}}$. 
    By the union bound over at most $n$ vertices with radial coordinates larger than $r'$, 
    the probability that no vertex with radial coordinate larger than $r'$ has inner degree larger than 
    than the size of the initial clique is at most $n2^{-c_3 n^{1-\alpha}}=2^{-\Omega(n^{1-\alpha})}$.
    In other words, with probability $1 - n2^{-c_3 n^{1-\alpha}} = 1-2^{-\Omega(n^{1-\alpha})}$,
    all vertices with radial coordinates larger than $r'$ have inner degree larger than the size
    of the initial clique. 

    If this event happens, we remove all vertices with coordinates larger than $r'$. 
    Therefore, it suffices to show that the number of vertices with coordinates at most $r'$ is small.
    The expected number of such vertices is $n\mu(B_0(r'))$.
    For this, we obtain the explicit formula for $r'$ as follows.
    From $n\mu (B_0(r') \cap B_{r'}(R)) = c_2 n^{1-\alpha}$, we get 
    \begin{align*}
        n\mu(B_0(r') \cap B_{r'}(R)) &= \frac{2\alpha}{\pi(\alpha-1/2)} ne^{-R(\alpha - 1/2) + r'(\alpha -1)} \cdot (1-o(1))\\
        &= \frac{c_3}{2e} n^{1-\alpha}.
    \end{align*}
    Applying $R = 2\ln n + C$, we obtain
    \begin{align*}
        e^{r'} = \left( \frac{c_3}{2e}\frac{\pi(\alpha-1/2)}{2\alpha}e^{C(\alpha-1/2)} \right)^{\frac{1}{\alpha-1}} n (1 + o(1)).
    \end{align*}
    Let $Y$ be the number of vertices lying inside $B_0(r')$. 
    Then the expected number of vertices we have after applying the deletion process is as follows. 
    \begin{align*}
        \mathbb{E}[Y] &= n\mu(B_0(r')) = ne^{-\alpha(R-r')}\\
        &= ne^{-\alpha(2\ln n + C)} \left( \frac{c_3}{2e}\frac{\pi(\alpha-1/2)}{2\alpha}e^{C(\alpha-1/2)} \right)^{\frac{\alpha}{\alpha-1}} n^\alpha (1 + o(1)).
    \end{align*}
    Therefore, 
    \begin{align*}
        \mathbb{E}[Y] \leq c_4 n^{1-\alpha},
    \end{align*}
    for some constant $c_4$ depending only on $\alpha$ and $C$.
    By Chernoff bound, 
    \begin{align*}
        \mathbb{P}[Y > 2ec_4 n^{1-\alpha}] \leq 2^{-2ec_4 n^{1-\alpha}}. 
    \end{align*}

    In summary, 
    with probability $1-2^{-\Omega(n^{1-\alpha})}$,
    all vertices with radial coordinates larger than $r'$ have inner degree larger than the size
    of the initial clique.
    Also, with probability $1- 2^{-2ec_4 n^{1-\alpha}}$,
    the number of vertices in $B_0(r')$ is at most $2ec_4 n^{1-\alpha}$. 
    Therefore, the probability that 
    the number of remaining vertices after applying the deletion process is at most $2ec_4 n^{1-\alpha}$
    is at least $1-2^{-2ec_4 n^{1-\alpha}}- 2^{-\Omega(n^{1-\alpha})}$, which is
    $1-2^{-\Omega(n^{1-\alpha})}$.
\end{proof}

Although the deletion process we use for analysis requires the geometric representation of $G$,
the original deletion process does not require the geometric representation of $G$. 
By combining the argument in Section~\ref{sec:initial}
and Section~\ref{sec:degree-reduction}, we have the following theorem. 


\begin{theorem}\label{thm:expected}
    Given a graph drawn from $\mathcal G_{n,\alpha,C}$ with $\frac{1}{2} < \alpha < 1$ and its geometric representation, 
    we can compute its maximum clique  in $\mathcal{O}\left(m+ n^{4.5(1-\alpha)} \right)$ expected time. 
\end{theorem}

\begin{proof}
    We first show 
       its maximum clique can be computed in $\mathcal{O}\left(m+ n^{4.5(1-\alpha)} \right)$ time with probability at least
    $1-2^{-\Omega(n^{1-\alpha})}$. 
    We can compute the initial solution of size $\Omega(n^{1-\alpha})$ in $O(n+m)$ time with probability 
    $1-2^{-\Omega(n^{1-\alpha})}$. 
    Then we apply the algorithm in~\cite{blasius2018cliques}
    to the initial solution 
    for computing a maximum clique of a hyperbolic random graph
    with $n'$ vertices and $m'$ edges 
    in $O(m'n'^{2.5})$ time. Since $m'=n'^2$ in the worst case, 
    the total running time is $O(m+n^{4.5(1-\alpha)})$. 
    
    Even in the case that
    the reduction fails to delete a sufficient number of vertices, the running time of the algorithm is $O(m+n^{4.5})$. 
    Thus the expected running time is 
    \begin{align*}
        O(m+n^{4.5(1-\alpha)}) (1-2^{-\Omega(n^{1-\alpha})}) + O(m+n^{4.5}) 2^{-\Omega(n^{1-\alpha})} = O(m+n^{4.5(1-\alpha)}).\tag*{\qedhere}
    \end{align*}
\end{proof}

\section{Efficient Robust Algorithm for the Maximum Clique Problem}\label{sec:robust-algo}
In this section, we present the first algorithm for the maximum clique problem on hyperbolic random graphs which does not require geometric representations. 
In many cases, a geometric representation of a graph is not given. 
In particular, real-world networks such as social  and collaboration networks are not defined based on geometry although they share properties with HRGs.
As we want to use hyperbolic random graphs as a model for such real-world networks, it is necessary to design algorithms not requiring geometric representations. 

Our main key tool in this section is
the notion of \emph{co-bipartite neighborhood edge elimination ordering} (CNEEO) introduced by Raghavan and Spinrad~\cite{raghavan2003robust}.
It can be considered as a variant of a perfect elimination ordering. 
Let $G$ be an undirected graph. 
Let $L = (e_1, e_2 , \ldots e_m)$ be an edge ordering of all edges of $G$.
Let $G_L[i]$ be the subgraph of $G$ with the edge set $\{ e_i, e_{i+1} \ldots e_m \}$.
For a vertex $v\in V$, let 
$N_{G}(v)$ denote the set of neighbors of $v$ in $G$. 
Then $L$ is called a \emph{co-bipartite neighborhood edge elimination ordering} (CNEEO)  if for each edge $e_i = (u_i, v_i)$, 
the subgraph of $G$ induced by 
$N_{G_L[i]}(u_i) \cap N_{G_L[i]}(v_i)$ is co-bipartite. 
Here, a \emph{co-bipartite graph} is a graph whose complement is bipartite. 

Raghavan and Spinrad~\cite{raghavan2003robust} presented an algorithm for computing a CNEEO in polynomial time if a given graph admits a CNEEO.
Moreover, they presented a polynomial-time algorithm for computing a maximum clique in polynomial time assuming a CNEEO is given. 
We first show that a hyperbolic unit disk graph admits a CNEEO. This immediately leads to a polynomial-time algorithm for the maximum clique problem. 

\begin{lemma}\label{lem:cneeo}
    Every hyperbolic unit disk graph admits a CNEEO.   
\end{lemma}
\begin{proof}
    Let $G$ be a hyperbolic unit disk graph, and $L = (e_1, e_2, \ldots, e_m)$ be the ordered set of all edges in $G$ sorted in the non-increasing order of their lengths. 
    For each edge $e_i = (u, v)$, 
    $d(u, w) \leq d(u, v)$ and $d(v, w) \leq d(u, v)$ for all $w \in N_{G_L[i]}(u) \cap N_{G_L[i]}(v)$.
    Thus the vertices in $ N_{G_L[i]}(u) \cap N_{G_L[i]}(v)$ are contained in $R_{uv}=B_u(r) \cap B_v(r)$ with $r=d(u,v)$.
    Because the subgraph of $G$ induced by the vertices in $R_{uv}$ is co-bipartite, 
    the subgraph of $G$ induced by 
    $N_{G_L[i]}(u_i) \cap N_{G_L[i]}(v_i)$ is also co-bipartite. 
    Thus, $L$ is the CNEEO.
\end{proof}

Since Raghavan and Spinrad~\cite{raghavan2003robust} did not give an explicit time bound
on their algorithm, we describe their algorithm and analyze their running time in the following lemma. 
Note that the following lemma holds for an arbitrary graph (not necessarily a hyperbolic unit disk graph) admitting a CNEEO. 
\begin{lemma}[\cite{raghavan2003robust}]
    Given a graph $G = (V, E)$ which admits a CNEEO,
    the maximum clique problem can be solved in $O(m^3 + mn^{2.5})$ time.
\end{lemma}
\begin{proof}
    If $G$ has a CNEEO, we can compute it in a greedy fashion in $O(m^3)$ time. 
    Starting with an empty ordering, we add edges to the ordering one by one.
    At each step $i$, we have an ordering $(e_1,\ldots,e_{i-1})$, and we extend this ordering
    by adding an edge. Let $G_i$ be the subgraph of $G$ induced by $E-\{e_1,\ldots,e_{i-1}\}$.
    Then for each edge $e$, we check if the common neighbors of its endpoints in $G_i$
    induces a co-bipartite graph in $G$. If so, we add it at the end of the current ordering.
    If $G$ has a CNEEO, Raghavan and Spinrad~\cite{raghavan2003robust} showed that
    this greedy algorithm always returns a CNEEO. 
    For its analysis, observe that we have $O(m)$ steps.
    For each step $i$, the number of candidates for $e_i$ is $O(m)$, and
    for each candidate $e$, we can check if the common neighbors $N$ of its endpoints induces
    a co-bipartite graph in $O(n+m)$ time. To see this, 
    observe that $G[N]$ cannot be co-bipartite if $|E(G[N])| < \frac{(|N|-1)^2}{2}$.
    Thus if the number of edges of $G[N]$ is at least $\frac{(|N|-1)^2}{2}$,
    we can immediately conclude that $e$ cannot be the $i$th edge $e_i$ in the ordering. 
    Otherwise, we can check if $G[N]$ is co-bipartite in $O(N^2)$ time, which is at most $O(m)$ time. In this way, we can complete each step in $O(m^2)$ time, and thus the total running time
     is $O(m^3)$. 
    
    Then using the CNEEO, we can compute a maximum clique as follows.
    For each $i$, we compute a maximum clique of the subgraph of $G$ induced by $N_{G_L[i]}(u) \cap N_{G_L[i]}(v)$. Since it is co-bipartite, it is equivalent to computing a
    maximum independent set of its complement, which is bipartite.
    This takes $O(n^{2.5})$ time using the Hopcroft-Karp algorithm. 
     Raghavan and Spinrad~\cite{raghavan2003robust} showed that
      a largest clique among all $n$ such cliques is a maximum clique of $G$.    
    Thus, we can find maximum clique of $G$ in $O(mn^{2.5})$.
\end{proof}

In the case of hyperbolic \emph{random} graphs, we can solve the problem even faster.
As we did in Section~\ref{sec:nonrobust-algo},
we compute an initial clique,  remove
vertices of small degrees, and then obtain a kernel of small size.
Recall that these procedures do not require geometric representations. 
Note that the kernel is also a hyperbolic unit disk graph because we remove vertices only. 
The number of vertices of the kernel is $O(n^{1-\alpha})$ and the edges is $O(n^{2-2\alpha})$ in probability $1-2^{-\Omega(n^{1-\alpha})}$.
Thus, we have the following theorem.


\begin{corollary}\label{thm:robust}
    Given a graph drawn from $\mathcal G_{n,\alpha,C}$ with $\frac{1}{2} < \alpha < 1$, 
    we can compute a maximum clique in $\mathcal{O}\left(m + n^{6(1-\alpha)} \right)$ expected time  without its geometric representation.
\end{corollary}

\subparagraph{Heuristics for real-world networks.} 
A main motivation of the study of hyperbolic random graphs is to obtain heuristics for 
analyzing real-world networks. 
Many real-world networks share salient features with hyperbolic random graphs,
but this does not mean that many real-world networks \emph{are} hyperbolic random graphs.
Because the algorithm in Corollary~\ref{thm:robust} is aborted for a graph not admitting a CNEEO,
one cannot expect that this algorithm works correctly for many real-world networks. 
In fact, only a few real-world networks admit CNEEO as we will see in Section~\ref{sec:experiments}. 
That is, for most of real-world networks, the algorithm in Corollary~\ref{thm:robust} is aborted. 

However, in this case, we can obtain a lower bound on the optimal clique sizes, and moreover, we can reduce the size of the graph. 

\begin{lemma}\label{lem:remaining}
    Assume that the algorithm in Theorem~\ref{thm:robust} is aborted at step $i$. 
    Let $(e_1,\ldots,e_{i-1})$ be the ordering we have at stage $i$. For an index $j\leq i$, let $G_j$ be the 
    subgraph of $G$ with edge set $E-\{e_1,\ldots,e_{j-1}\}$.
    Then either  
     a maximum clique of $G_i$ is a maximum clique of $G$, or 
     a maximum clique of the subgraph of $G$ induced by the common neighbors of the endpoints of $e_j$ tin $G_j$ 
     is a maximum clique of $G$ for some index $j<i$. 
\end{lemma}
\begin{proof}
    Let $K$ be a maximum clique of $G$. 
    If an edge of the ordering we have at stage $i$
    has both endpoints in $K$,
    let $e_j$ be the first edge in the ordering among them. 
    In this case, $K$ is a clique of 
    the subgraph $G'$ induced by
    the common neighbors of the endpoints of $e_j$ in $G_j$.
    To see this, 
    observe that all vertices of $K$ are common neighbors of 
    the endpoints of $e_j$, and they are vertices in $G_j$.
    Therefore, a maximum clique of $G'$ has size $|K|$, and thus the lemma follows.

    If no edge of the ordering we have at stage $i$
    has both endpoints in $K$,
    then $K$ is a clique in $G_i$. Therefore, it belongs to the first case described in the statement of the lemma.
    Therefore, the lemma holds for any case. 
\end{proof}

Although we do not have any theoretical bound here, our experiments showed that 
the size of the clique we can obtain is close to the optimal value for many instances.
Details will be described in Section~\ref{sec:experiments}.

\section{Improving Performance through Additional Optimizations}\label{sec:opt}
For implementation, we introduce the following two minor techniques for improving the performance of the algorithms. Although these techniques do not improve the performance theoretically, they improve the performance empirically. 
Recall that our algorithms consist of two phases:
Computing a kernel of size $O(n^{1-\alpha})$, 
and then computing a maximum clique of the kernel.
As the first phase can be implemented efficiently,
we focus on the second phase here.
Again, the second phase has two steps.
With geometric representations, we first compute a CNEEO, and then compute a maximum clique using the CNEEO.
Without geometric representations, we consider every edge, and then compute a maximum clique in the subgraph induced by
the common neighbors of the endpoint of the edge.
The first technique applies to both of the two steps, and 
the second technique applies to the first step. 

\subsection{Skipping Vertices with Low Degree}\label{sec:skip-low}
The main observation of our kernelization algorithm is that, for any lower bound $k$ on the size of a maximum clique, 
a vertex of degree less than $k$ does not participate in a maximum clique. 
The first technique we use in the implementation is to make use of this observation
also for computing a CNEEO, and for computing a maximum clique using the CNEEO. 

While computing a CNEEO, the lower bound $k$ we have does not change; it is the size of the initial clique.
Whenever we access a vertex, we check if its degree is less than $k$.
If so, we remove this vertex from the kernel, and do not consider it any more.
One can consider this as ``lazy deletion.''
Then once we have a CNEEO, we scan the edges in the CNEEO, and for each edge,
we compute a maximum clique of the subgraph defined by the edge.
If it is larger than the lower bound $k$ we have, we update $k$ accordingly.
In this process, whenever we find a vertex of degree less than $k$, we remove it immediately.
Moreover, if the subgraph defined by each edge of CNEEO has vertices less than $k$,
we skip this subgraph as it does not contain a clique of size larger than $k$. 



\subsection{Introducing the Priority of Edges}\label{sec:edge-priority} 
Recall that the second phases consists of two steps: computing a CNEEO, and computing a maximum clique using the CNEEO. In this section, we focus on the first step.

\subparagraph{With Geometric Representations.}
In this case, we use the $O(m'n'^{2.5})$-time algorithm by~\cite{blasius2018cliques} for computing a maximum clique of the kernel with $n'$ vertices and $m'$ edges. 
Although it is the theoretically best-known algorithm, we observed that
computing a maximum clique using a CNEEO is more efficient practically.
By the proof in lemma~\ref{lem:cneeo}, the list of
the edges sorted in the non-decreasing order of their lengths 
is a CNEEO. 
Without using a CNEEO, for each edge $uv$, we have to compute the subgraph of $G$  induced by the common neighbors of $u$ and $v$ in $G$.
On the other hand, once we have a CNEEO, it suffices to consider the subgraph of $G$ induced by the 
common neighbors of $u$ and $v$ in $G_L$, where $G_L$ is the subgraph of $G$ with the edges coming after $uv$ in the CNEEO. 
If $uv$ lies close to the last edge in the CNEEO,
the number of common neighbors of $u$ and $v$ in $G_L$ can be significantly
smaller than the number of their common neighbors in $G$. 
This can lead to the performance improvement. 


\subparagraph{Without Geometric Representations.}
In this case, we compute a CNEEO in a greedy fashion. Starting from the empty sequence, we add the edges one by one in order. For each edge $e$ not added to the current ordering, we check if the common neighbors of the endpoints of $e$ in the kernel is co-bipartite. If an edge passes this test, we add it to the ordering.
It is time-consuming especially when only a few edges can pass the test.
To avoid considering the same edge repeatedly, we use the following observation. 
Once an edge $e$ fails this test, it cannot pass the test unless one of its incident edges are added to the ordering. 
Using this observation, we classify the edges into two sets: active edges and inactive edges. 
In each iteration, we consider the active edges only. Once an edge fails the test, then it becomes inactive.
Once an edge passes the test, we make all its incident edges active. 
In this way, we 
can significantly improve the running time especially for graphs that do not accept CNEEO.

\section{Experimental Evaluation}\label{sec:experiments}
In this section, we evaluate the performance of our algorithm mainly on 
hyperbolic random graphs and real-world networks.

\subparagraph{Environment and data.}
We implemented our algorithm using C++17.
The code were compiled with GNU GCC version $11.3.0$ with optimization flag "-O2".
All tests were run on a desktop with Rygen 7 3800X CPU, 32GB memory, and Ubuntu 22.04LTS.

We evaluate the performance of our algorithm on hyperbolic random graphs and real-world networks. 
For hyperbolic random graphs, we generate graphs using  the open source library GIRGs~\cite{blasius2022efficiently} by setting parameters differently. 
Recall that we have three parameters $n$, $C$ and $\alpha$. 
Here, instead of $C$, we use the average degree, denoted by $\delta$, as a parameter  because $\delta$ can be represented as a function of $C$ and $\alpha$. 
As we consider the average performance of our algorithm, we sampled
100 random graphs for 
fixed parameters $n, \delta$ and $\alpha$, and then calculate the average results (the size of kernels or the running times).

For real-world networks, we use the SNAP dataset~\cite{snapnets}.
It contains directed graphs and non-simple graphs as well. In this case, 
we simply ignore the directions of the edges and interpret all directed graphs  as undirected graphs. 
Also, we collapse all multiple edges into a single edge and remove all loops.


\subsection{Experiment on Hyperbolic Random Graphs: Kernel Size}
\begin{figure}
    \centering
    \includegraphics[width=\linewidth]{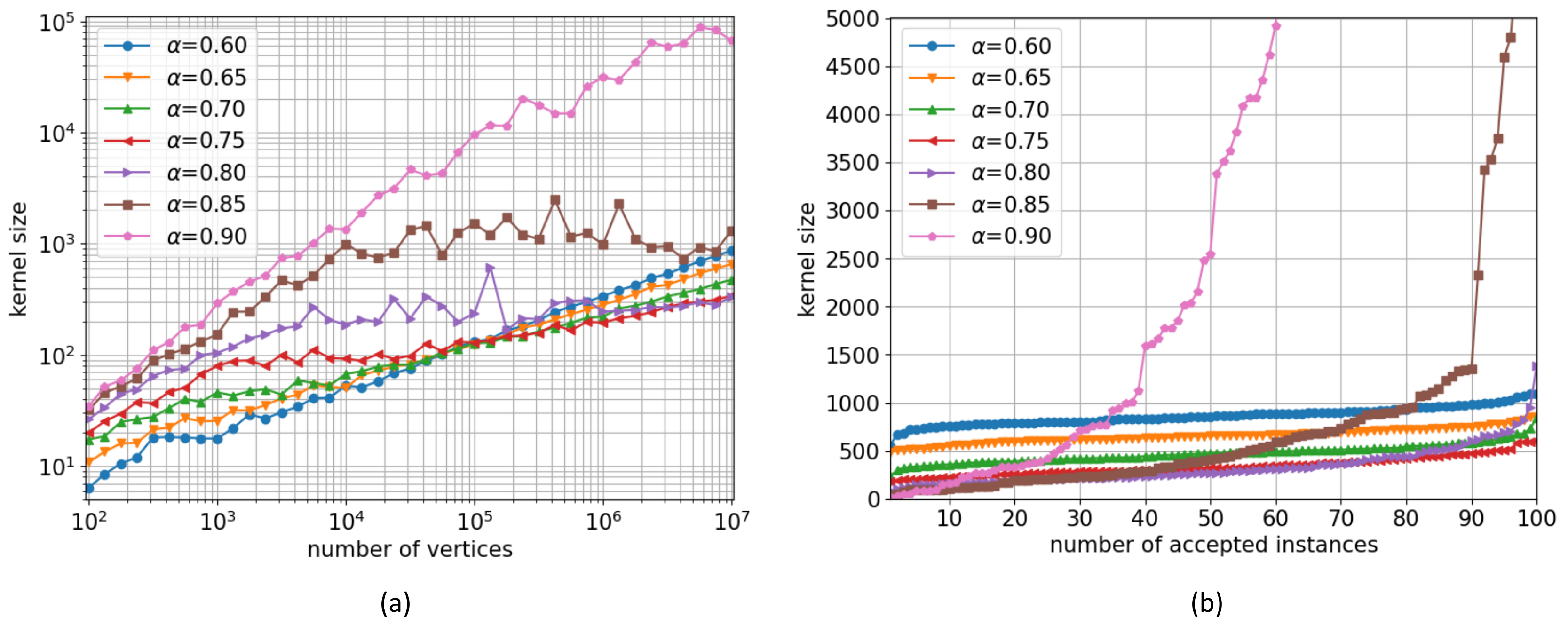} 
    \caption{(a) Comparison of the kernel sizes with varying $\alpha$. 
    Here, $\delta=10$. (b) Cactus plot of the kernel size versus the number of accepted instances for each value of $\alpha$ with fixed $n=10^7$.}
    \label{fig:reduction_alpha}
\end{figure}

We showed that the size of the kernel of the hyperbolic random graph is $O(n^{1-\alpha})$ with probability $1-2^{\Omega(n^{1-\alpha})}$.
In this section, 
we evaluated the tendency on the size of the kernel experimentally 
as $n$, $\alpha$ and $\delta$ change. Here, $\alpha$ controls the power-law exponent, and $\delta$ is the average degree of the graph.

Figure~\ref{fig:reduction_alpha} shows the tendency on the size of the kernel as $\alpha$ changes. 
Here, we fix $\delta=10$. 
Figure~\ref{fig:reduction_alpha}(a) shows a plot of the kernel size versus the number of vertices of a graph on a log–log scale for each value of $\alpha$. 
We generate 100 instances randomly and take the average of their results for each point in the plot.
Figure~\ref{fig:reduction_alpha}(b) shows a cactus plot of the kernel size versus the number of accepted instances for each value of $\alpha$. Here, for a fixed kernel size $k$,
an instance is said to \emph{accepted} if our algorithm returns a kernel of size at most $k$ for this instance. 
Here, we fix $n=10^7$ and $\delta=10$. 


For $\alpha \leq 0.8$, the size of the average kernel decreases for sufficiently large $n$, say $n=10^7$, as $\alpha$ increases in Figure~\ref{fig:reduction_alpha}(a).
Also, the kernel sizes for all instances are concentrated on the average kernel size for each $\alpha \leq 0.8$ in Figure~\ref{fig:reduction_alpha}(b).
This is consistent with Lemma~\ref{lem:kernel-size} stating that 
the kernel size is $O(n^{1-\alpha})$ with high probability. 
However, this fact does not hold for $\alpha > 0.8$ in Figure~\ref{fig:reduction_alpha}(a), 
and the reason for this can be seen in Figure~\ref{fig:reduction_alpha}(b). 
At $\alpha = 0.85$, approximately $10\%$ of instances did not have kernels of size at most $1500$, and at $\alpha = 0.9$, over $60\%$ of instances did not have such kernels. 
Notice that the plot sharply increases when the kernel size exceeds $1500$.
The success probability stated in  Lemma~\ref{lem:kernel-size} is
$1 - 2^{\Omega(n^{1-\alpha})}$, 
which decreases as $\alpha$ increases.
In other words, if $n$ is not sufficiently large, 
it is possible that the success probability $1 - 2^{\Omega(n^{1-\alpha})}$ is not sufficiently large for $\alpha>0.8$. 
That is, if we increase the number of vertices on our experiments, we would 
get the desired tendency on the kernel size for all values $\alpha$.
Nevertheless, at $n=10^7$, our algorithm removes 
a significant number of vertices, 
leaving only $0.01\%$ of vertices at $\alpha=0.85$ and only $1\%$ at $\alpha=0.9$.

\begin{figure}
    \centering
    \includegraphics[width=\linewidth]{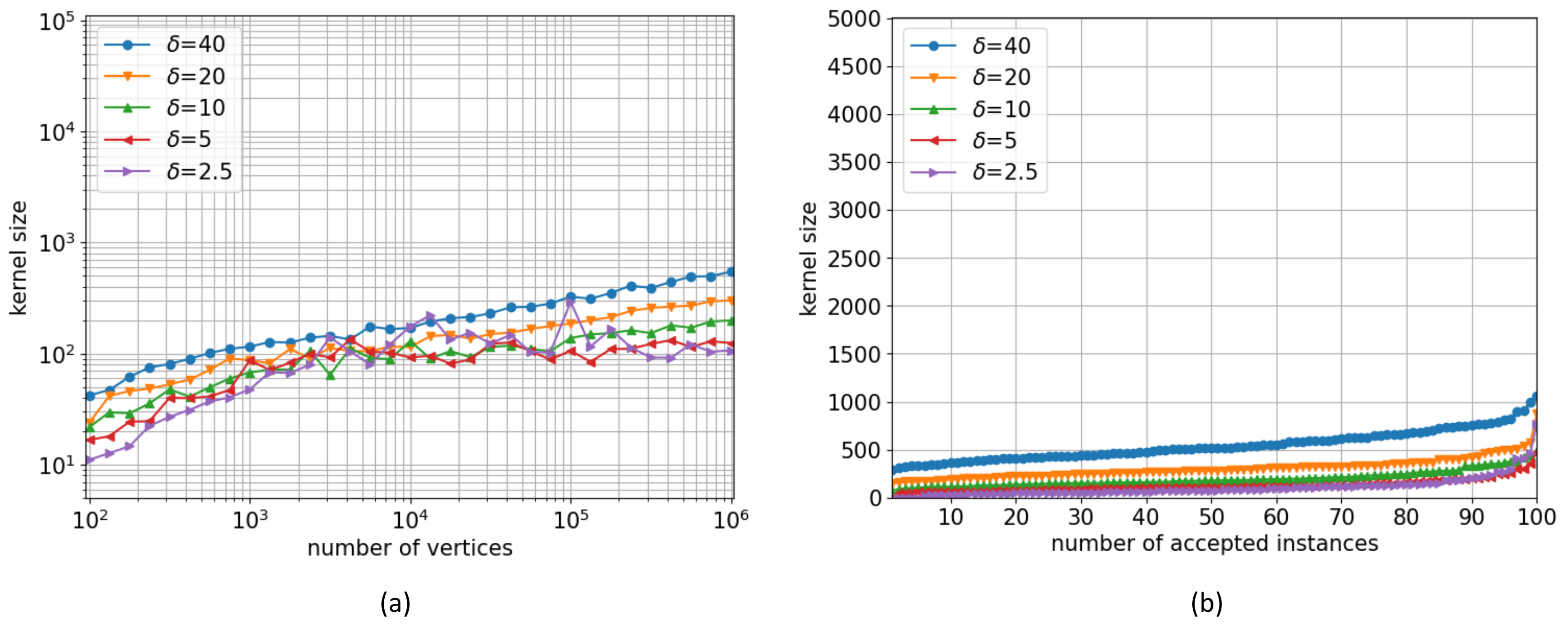}
    \caption{(a) Comparison of the kernel sizes with varying  $\delta$ for fixed 
    $\alpha=0.75$. (b) Cactus plot of the kernel size versus the number of accepted instances for each value of $\delta$ and fixed $n=10^6$.}
    \label{fig:reduction_C}
\end{figure}
Figure~\ref{fig:reduction_C}(a) shows 
that the kernel size increases as $\delta$ increases.
Here, we fix $\alpha=0.75$. 
Figure~\ref{fig:reduction_C}(b) shows a cactus plot of the kernel size versus the number of accepted instances for each value of $\delta$.
The kernel sizes for all instances are concentrated on the average kernel size for each $\delta$. 
This tendency can be explained 
from the proof of 
Lemma~\ref{lem:kernel-size}. 
The constant hidden behind the O-notation in the lemma depends on $C$, and thus 
it also depends on $\delta$. 
By carefully analyzing this constant, we can show that  the average kernel size 
increases as $\delta$  increases (for a fixed $\alpha<1$.) 
Similarly, 
we can show that the success probability 
increases as $\delta$ increases.

\subsection{Experiment on Hyperbolic Random Graphs: Running Time}
\begin{figure}
    \centering
    \includegraphics[width=\linewidth]{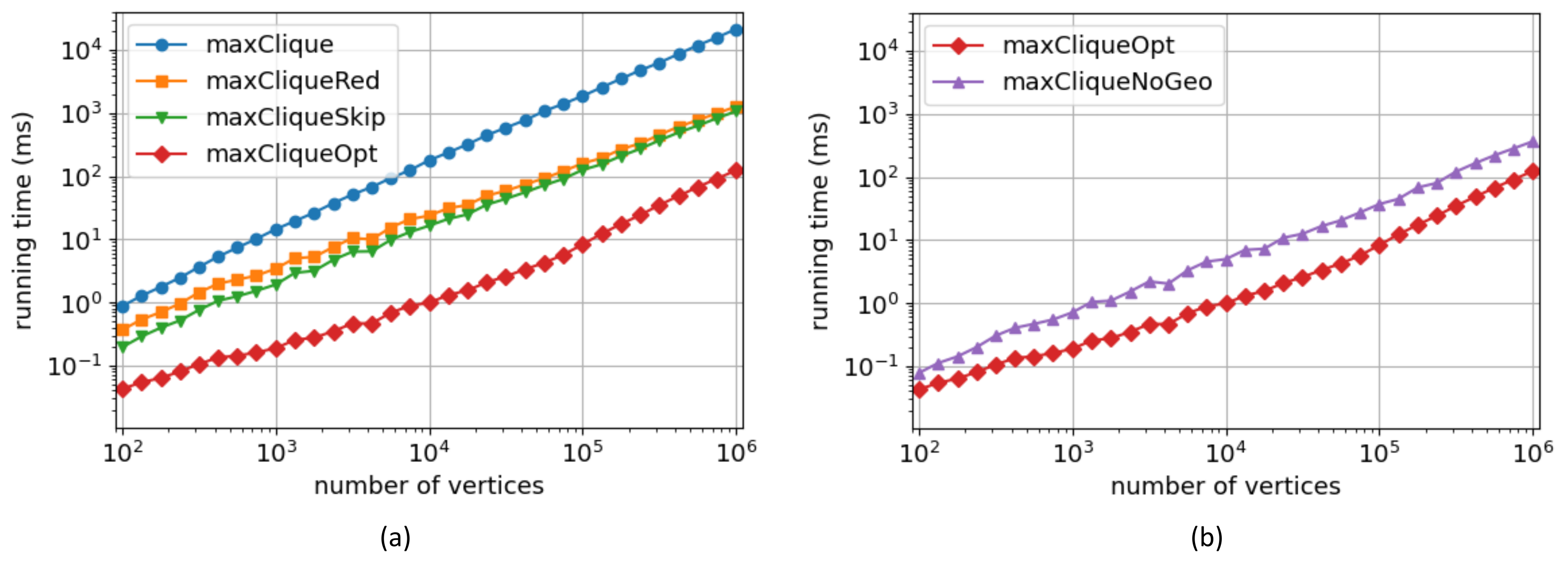}
    \caption{(a) Comparison of running times of different versions of our algorithm. (b) Comparison of running times with and without geometric representations. Here, $\alpha=0.75$ and $\delta=10$.}
    \label{fig:running_time}
\end{figure}

\begin{table}[]
\centering
\begin{tabular}{|l|rrcrrr|r|}
\hline
\multicolumn{1}{|c|}{} &
  \multicolumn{1}{c}{\textsf{INIT}} &
  \multicolumn{1}{c}{\textsf{KERNEL}} &
  \textsf{CNEEO} &
  \multicolumn{1}{c}{\textsf{CONST}} &
  \multicolumn{1}{c}{\textsf{INDEP}} &
  \multicolumn{1}{c|}{\textsf{OTHER}} &
  \multicolumn{1}{c|}{\textsf{TOTAL}} \\ \hline
\texttt{MaxClique}        & \multicolumn{1}{c}{-} & \multicolumn{1}{c}{-} & -                          & 21\ 516.34 & 4\ 470.38 & 8.18 & 25\ 994.90 \\ \hline
\texttt{MaxCliqueRed}     & 15.67                 & 89.27                 & -                          & 1\ 126.95  & 37.92     & 2.88 & 1\ 272.70  \\ \hline
\texttt{MaxCliqueSkip}    & 15.59                 & 88.62                 & -                          & 925.33     & 30.77     & 2.88 & 1\ 063.19  \\ \hline
\texttt{MaxCliqueOpt}     & 15.64                 & 88.61                 & \multicolumn{1}{r}{1.07}   & 12.55      & 1.19      & 2.88 & 121.94     \\ \hline
\texttt{MaxCliqueNoGeo}               & 15.32                 & 89.25                 & \multicolumn{1}{r}{258.45} & 5.82       & 0.93      & 2.96 & 372.74     \\ \hline
\end{tabular}
\caption{
Running time for each operation.
The unit of time is a millisecond. \vspace{-1em}}
\label{table:operation_time}
\end{table}

In this section, we conducted  experiments for evaluating the running times of different versions of our algorithms.
Figure~\ref{fig:running_time} shows a plot of the number of vertices versus the running time of each version of our algorithm. 
Here, we fix $\alpha=0.75$ and $\delta=10$.
Each point of the plot is averaged for $100$ instances.  
Figure~\ref{fig:running_time}(a) shows a plot for 
the algorithm, denoted by $\texttt{MaxClique}$, 
by Bl\"asius et al.~\cite{blasius2018cliques} and 
three different versions of the  algorithm using geometric representations: $\texttt{MaxCliqueRed}$, $\texttt{MaxCliqueSkip}$, 
and $\texttt{MaxCliqueOpt}$.
More specifically, $\texttt{MaxCliqueRed}$ denotes the algorithm described in Section~\ref{sec:robust-algo}.
$\texttt{MaxCliqueSkip}$ denotes the algorithm described in Section~\ref{sec:skip-low} that  
skips low-degree vertices.
$\texttt{MaxCliqueOpt}$ denotes the algorithm described in Section~\ref{sec:edge-priority} 
that introduces priorities of edges. 
As expected, $\texttt{MaxCliqueOpt}$ outperforms all other versions of the algorithms in this experiment.

For a precise analysis, we evaluated the running time of each task for our algorithms and reported them in 
in Table~\ref{table:operation_time}. 
More specifically, the algorithms conduct six tasks:  
\textsf{INIT} denotes
the task of finding an initial solution.
\textsf{KERNEL} denotes the task of finding the kernel.
\textsf{CNEEO} denotes the task of computing a CNEEO.
\textsf{CONST} denotes the task of constructing a co-bipartite graph by considering the common neighbors of the endpoints of each edge.
\textsf{INDEP} denotes the task of  computing a maximum independent set of the complement of a co-bipartite graph. 
\textsf{OTHER} denotes all the other tasks such as the initialization for variables and caches.
\textsf{TOTAL} denotes the entire tasks of our algorithm. 

As expected, 
$\texttt{MaxCliqueRed}$ outperforms
$\texttt{MaxClique}$ significantly. 
However, \textsf{CONST} is still a time-consuming task for $\texttt{MaxCliqueRed}$. 
Thus we focus on optimization techniques for \textsf{CONST} and provides  $\texttt{MaxCliqueSkip}$ and $\texttt{MaxCliqueOpt}$ in Section~\ref{sec:opt}. 
Although $\texttt{MaxCliqueSkip}$
gives a performance improvement, it still takes a significant amount of time in \textsf{CONST} and \textsf{INDEP}. 
$\texttt{MaxCliqueOpt}$ computes a CNEEO by sorting the edges with respect to their lengths. This allows us to manage degree efficiently and apply low-degree skip technique to a larger number of vertices.
This significantly improves the running time of $\texttt{MaxCliqueSkip}$ for \textsf{CONST} and \textsf{INDEP}.
This algorithm runs in about $100$ms even at $n=10^6$, 
exhibiting a performance improvement 
over $200$ times compared to $\texttt{MaxClique}$.

Next, we compared the running times of two algorithms with and without geometrical representations. 
In the case that a geometric representation is given, we use $\texttt{MaxCliqueOpt}$.
If a geometric representation is not given, we use the algorithm in Section~\ref{sec:robust-algo} and denote it by  $\texttt{MaxCliqueNoGeo}$.
In $\texttt{MaxCliqueOpt}$, we can quickly compute a CNEEO by sorting the edges in non-decreasing order of length. 
However, $\texttt{MaxCliqueNoGeo}$ 
computes a CNEEO in a greedy approach which incurs significant overhead. 
Despite this, the performance of $\texttt{MaxCliqueNoGeo}$ in Figure~\ref{fig:running_time}(b) does not show a significant difference compared to $\texttt{MaxCliqueOpt}$, and it even outperforms $\texttt{MaxCliqueSkip}$.

\begin{table}[]
\centering
\resizebox{\textwidth}{!}{%
\begin{tabular}{|l|r|r|r|r|r|r|r|r|r|}
\hline
 &
  \multicolumn{1}{c|}{$|V|$} &
  \multicolumn{1}{c|}{$|E|$} &
  \multicolumn{1}{c|}{$|V_\text{kernel}|$} &
  \multicolumn{1}{c|}{$|V_\text{left}|$} &
  \multicolumn{1}{c|}{$|E_\text{left}|$} &
  \multicolumn{1}{c|}{runtime} &
  \multicolumn{1}{c|}{$\omega_\text{kernel}$} &
  \multicolumn{1}{c|}{$\omega_\text{eval}$} &
  \multicolumn{1}{c|}{$\omega$} \\ \hline
as-skitter    & 1\ 696,415  & 11\ 095,298  & 28\ 787     & 17\ 033  & 693\ 272    & 342.63    & 37  & $\geq$ 63  & 67  \\ \hline
ca-AstroPh    & 18\ 771     & 198\ 050     & 3\ 679      & 0        & 0           & 1.18      & 23  & 57         & 57  \\ \hline
ca-CondMat    & 23\ 133     & 93\ 439      & 13\ 464     & 0        & 0           & 0.07      & 4   & 26         & 26  \\ \hline
ca-HepPh      & 11\ 204     & 118\ 489     & 0           & 0        & 0           & 0.00      & 239 & 239        & 239 \\ \hline
com-amazon    & 334\ 863    & 925\ 872     & 255\ 473    & 0        & 0           & 0.55      & 3   & 7          & 7   \\ \hline
com-dblp      & 317\ 080    & 1\ 049\ 866  & 1\ 716      & 0        & 0           & 0.21      & 26  & 114        & 114 \\ \hline
com-lj        & 3\ 997\ 962 & 34\ 681\ 189 & 1\ 713\ 237 & 126\ 388 & 4\ 587\ 418 & 2\ 316.20 & 7   & $\geq$ 289 & 327 \\ \hline
com-youtube   & 1\ 134\ 890 & 2\ 990\ 443  & 36\ 716     & 7\ 919   & 211\ 051    & 53.48     & 13  & $\geq$ 14  & 17  \\ \hline
Gnutella31    & 62\ 586     & 147\ 892     & 33\ 816     & 0        & 0           & 0.05      & 2   & 4          & 4   \\ \hline
Slashdot0811  & 77\ 360     & 469\ 180     & 14\ 315     & 1\ 503   & 40\ 418     & 7.94      & 10  & $\geq$ 17  & 26  \\ \hline
Slashdot0902  & 82\ 168     & 504\ 229     & 13\ 964     & 1\ 543   & 42\ 215     & 9.24      & 11  & $\geq$ 17  & 27  \\ \hline
soc-Epinions1 & 75\ 879     & 405\ 740     & 9\ 337      & 3\ 717   & 148\ 354    & 26.09     & 10  & $\geq$ 22  & 23  \\ \hline
soc-pokec     & 1\ 632\ 803 & 22\ 301\ 964 & 1\ 252\ 317 & 54\ 101  & 924\ 531    & 288.37    & 4   & $\geq$ 29  & 29  \\ \hline
web-BerkStan  & 685\ 230    & 6\ 649\ 470  & 27\ 058     & 25\ 593  & 589\ 913    & 1\ 224.22 & 18  & $\geq$ 201 & 201 \\ \hline
web-Google    & 875\ 713    & 4\ 322\ 051  & 193\ 406    & 2\ 068   & 20\ 426     & 8.60      & 10  & $\geq$ 44  & 44  \\ \hline
web-NotreDame & 325\ 729    & 1\ 090\ 108  & 51\ 227     & 760      & 8\ 181      & 4.86      & 6   & $\geq$ 155 & 155 \\ \hline
web-Stanford  & 281\ 903    & 1\ 992\ 636  & 32\ 123     & 10\ 721  & 249\ 272    & 177.79    & 18  & $\geq$ 61  & 61  \\ \hline
WikiTalk      & 2,394\ 385  & 4\ 659\ 563  & 70\ 130     & 10\ 421  & 520\ 338    & 447.69    & 7   & $\geq$ 16  & 26  \\ \hline
Wiki-Vote     & 7\ 115      & 100\ 762     & 2\ 913      & 1\ 802   & 62\ 893     & 6.34      & 9   & $\geq$ 13  & 17  \\ \hline
\end{tabular}%
}
\caption{
The performance of our algorithm on the real-world data. 
}\vspace{-3mm}
\label{table:realworld}
\end{table}

\subsection{Experiment on Real-World Dataset}
Our algorithm can heuristically find large cliques for real-world data.
We conducted experiments on several real-world datasets and recorded these results in Table~\ref{table:realworld}.
The unit of the running time is a second.
$|V|$ and $|E|$ denote the numbers of vertices and edges of the input graph, respectively. 
$|V_\text{kernel}|$ denotes the number of vertices of the kernel. 
$|V_\text{left}|$ and $|E_\text{left}|$ denote the numbers of vertices and edges of the remaining graph described in Lemma~\ref{lem:remaining}, respectively.
Also, $\omega_\text{kernel}$ denotes the size of the initial clique,  $\omega_\text{eval}$ denotes the size of the clique computed from our algorithm, and $\omega$ denotes the size of the maximum clique of the graph.
Here, $\omega$ is the correct answer given by the dataset. 
If a given graph accepts a CNEEO, it is theoretically guaranteed that $\omega_\text{eval}$ is the exact solution and $|V_\text{left}| = |E_\text{left}| = 0$.
Otherwise, $\omega_\text{eval}$ is a lower bound
on the exact solution by Lemma~\ref{lem:remaining}, and $G_{\text{left}} = (V_\text{left}, E_\text{left})$ has a maximum clique if $\omega_\text{eval}$ is strictly smaller than the exact solution.

The collaboration networks such as ca-AstroPh, CondMat, HepPh, and com-dblp are one of the well-known scale-free networks. 
These networks accept a CNEEO, allowing us to find the exact maximum clique. 
Moreover, we were able to find a CNEEO considerably faster for these networks than for other graphs in our experiments.
Web graphs such as web-BerkStan, web-Google, web-Notre Dame, and web-Stanford are also one of the well-known scale-free networks. 
Although these graphs do not accept a CNEEO, we were able to reduce the number of vertices and edges significantly, 
and we obtained maximum cliques. 
For the other graphs we tested, we were able to obtain lower bounds that were close to the maximum clique size in most cases, 
and we were able to significantly reduce the size of the graphs.

\section{Conclusion}
We presented improved algorithms for the maximum clique problem on hyperbolic random graphs.
Our algorithms find a sufficiently large initial solution and find a sufficiently small kernel in linear time, which greatly improves the average time complexity and practical running time.
Also we gave the first algorithm for the maximum clique problem on hyperbolic random graphs without geometrical representations. 
Beyond the hyperbolic random graph, we applied these algorithms to real-world dataset and obtained lower bounds close to the optimum solutions for most of instances.

There are two possible directions for further improvement on our algorithms. 
First, we compute a maximum clique of hyperbolic random graphs 
using the framework for computing a maximum clique of unit disk graphs in~\cite{CLARK1990165}.
Recently, Espenant et al.~\cite{espenant2023finding} improved the algorithm~\cite{CLARK1990165} 
and presented an $O(n^{2.5} \log n)$-time algorithm for the maximum clique problem on unit disk graphs. It would be interesting if this technique  can be applied to 
hyperbolic geometry. 
Second, the bottleneck of our algorithm lies in constructing a CNEEO.
Especially, for most of real-world dataset, 
most of the running time is devoted to constructing a CNEEO. 
Thus to speed up the overall performance, this step must be improved. 

\bibliographystyle{plain}
\bibliography{paper}

\begin{thebibliography}{10}

\bibitem{abello2002massive}
James Abello, Mauricio~GC Resende, and Sandra Sudarsky.
\newblock Massive quasi-clique detection.
\newblock In {\em LATIN 2002: Theoretical Informatics: 5th Latin American
  Symposium Cancun, Mexico, April 3--6, 2002 Proceedings 5}, pages 598--612.
  Springer, 2002.

\bibitem{alimonti2000some}
Paola Alimonti and Viggo Kann.
\newblock Some {APX}-completeness results for cubic graphs.
\newblock {\em Theoretical Computer Science}, 237(1-2):123--134, 2000.

\bibitem{anderson2006hyperbolic}
James~W Anderson.
\newblock {\em Hyperbolic geometry}.
\newblock Springer Science \& Business Media, 2006.

\bibitem{blasius2022efficient}
Thomas Bl{\"a}sius, Cedric Freiberger, Tobias Friedrich, Maximilian Katzmann,
  Felix Montenegro-Retana, and Marianne Thieffry.
\newblock Efficient shortest paths in scale-free networks with underlying
  hyperbolic geometry.
\newblock {\em ACM Transactions on Algorithms (TALG)}, 18(2):1--32, 2022.

\bibitem{blasius2022efficiently}
Thomas Bl{\"a}sius, Tobias Friedrich, Maximilian Katzmann, Ulrich Meyer, Manuel
  Penschuck, and Christopher Weyand.
\newblock Efficiently generating geometric inhomogeneous and hyperbolic random
  graphs.
\newblock {\em Network Science}, 10(4):361--380, 2022.

\bibitem{blasius2016hyperbolic}
Thomas Bl{\"a}sius, Tobias Friedrich, and Anton Krohmer.
\newblock Hyperbolic random graphs: Separators and treewidth.
\newblock In {\em 24th Annual European Symposium on Algorithms (ESA 2016)}.
  Schloss Dagstuhl-Leibniz-Zentrum fuer Informatik, 2016.

\bibitem{blasius2018cliques}
Thomas Bl{\"a}sius, Tobias Friedrich, and Anton Krohmer.
\newblock Cliques in hyperbolic random graphs.
\newblock {\em Algorithmica}, 80(8):2324--2344, 2018.

\bibitem{blasius2018efficient}
Thomas Bl{\"a}sius, Tobias Friedrich, Anton Krohmer, and S{\"o}ren Laue.
\newblock Efficient embedding of scale-free graphs in the hyperbolic plane.
\newblock {\em IEEE/ACM transactions on Networking}, 26(2):920--933, 2018.

\bibitem{blasius2021efficiently}
Thomas Bl{\"a}sius, Tobias Friedrich, and Christopher Weyand.
\newblock Efficiently computing maximum flows in scale-free networks.
\newblock In {\em 29th Annual European Symposium on Algorithms (ESA 2021)}.
  Schloss Dagstuhl-Leibniz-Zentrum f{\"u}r Informatik, 2021.

\bibitem{bode2015largest}
Michel Bode, Nikolaos Fountoulakis, and Tobias M{\"u}ller.
\newblock On the largest component of a hyperbolic model of complex networks.
\newblock {\em The Electronic Journal of Combinatorics}, pages P3--24, 2015.

\bibitem{candellero2014clustering}
Elisabetta Candellero and Nikolaos Fountoulakis.
\newblock Clustering and the hyperbolic geometry of complex networks.
\newblock In {\em Algorithms and Models for the Web Graph: 11th International
  Workshop, WAW 2014, Beijing, China, December 17-18, 2014, Proceedings 11},
  pages 1--12. Springer, 2014.

\bibitem{CLARK1990165}
Brent~N. Clark, Charles~J. Colbourn, and David~S. Johnson.
\newblock Unit disk graphs.
\newblock {\em Discrete Mathematics}, 86(1):165--177, 1990.

\bibitem{dubhashi2009concentration}
Devdatt~P Dubhashi and Alessandro Panconesi.
\newblock {\em Concentration of measure for the analysis of randomized
  algorithms}.
\newblock Cambridge University Press, 2009.

\bibitem{espenant2023finding}
Jared Espenant, J.~Mark Keil, and Debajyoti Mondal.
\newblock Finding a maximum clique in a disk graph.
\newblock In Erin~W. Chambers and Joachim Gudmundsson, editors, {\em 39th
  International Symposium on Computational Geometry, SoCG 2023, June 12-15,
  2023, Dallas, Texas, {USA}}, volume 258 of {\em LIPIcs}, pages 30:1--30:17.
  Schloss Dagstuhl - Leibniz-Zentrum f{\"{u}}r Informatik, 2023.

\bibitem{friedrich2018diameter}
Tobias Friedrich and Anton Krohmer.
\newblock On the diameter of hyperbolic random graphs.
\newblock {\em SIAM Journal on Discrete Mathematics}, 32(2):1314--1334, 2018.

\bibitem{gugelmann2012random}
Luca Gugelmann, Konstantinos Panagiotou, and Ueli Peter.
\newblock Random hyperbolic graphs: degree sequence and clustering.
\newblock In {\em Automata, Languages, and Programming: 39th International
  Colloquium, ICALP 2012, Warwick, UK, July 9-13, 2012, Proceedings, Part II
  39}, pages 573--585. Springer, 2012.

\bibitem{kiwi2014bound}
Marcos Kiwi and Dieter Mitsche.
\newblock A bound for the diameter of random hyperbolic graphs.
\newblock In {\em 2015 Proceedings of the Twelfth Workshop on Analytic
  Algorithmics and Combinatorics (ANALCO)}, pages 26--39. SIAM, 2014.

\bibitem{doi:10.1137/18M121201X}
Marcos Kiwi and Dieter Mitsche.
\newblock On the second largest component of random hyperbolic graphs.
\newblock {\em SIAM Journal on Discrete Mathematics}, 33(4):2200--2217, 2019.

\bibitem{snapnets}
Jure Leskovec and Andrej Krevl.
\newblock {SNAP Datasets}: {Stanford} large network dataset collection.
\newblock \url{http://snap.stanford.edu/data}, June 2014.

\bibitem{lu2017finding}
Can Lu, Jeffrey~Xu Yu, Hao Wei, and Yikai Zhang.
\newblock Finding the maximum clique in massive graphs.
\newblock {\em Proceedings of the VLDB Endowment}, 10(11):1538--1549, 2017.

\bibitem{papadopoulos2010greedy}
Fragkiskos Papadopoulos, Dmitri Krioukov, Mari{\'a}n Bogun{\'a}, and Amin
  Vahdat.
\newblock Greedy forwarding in dynamic scale-free networks embedded in
  hyperbolic metric spaces.
\newblock In {\em 2010 Proceedings IEEE Infocom}, pages 1--9. IEEE, 2010.

\bibitem{pattabiraman2013fast}
Bharath Pattabiraman, Md~Mostofa~Ali Patwary, Assefaw~H Gebremedhin, Wei-keng
  Liao, and Alok Choudhary.
\newblock Fast algorithms for the maximum clique problem on massive sparse
  graphs.
\newblock In {\em Algorithms and Models for the Web Graph: 10th International
  Workshop, WAW 2013, Cambridge, MA, USA, December 14-15, 2013, Proceedings
  10}, pages 156--169. Springer, 2013.

\bibitem{raghavan2003robust}
Vijay Raghavan and Jeremy Spinrad.
\newblock Robust algorithms for restricted domains.
\newblock {\em Journal of algorithms}, 48(1):160--172, 2003.

\bibitem{robson1986algorithms}
John~Michael Robson.
\newblock Algorithms for maximum independent sets.
\newblock {\em Journal of Algorithms}, 7(3):425--440, 1986.

\bibitem{rossi2014fast}
Ryan~A Rossi, David~F Gleich, Assefaw~H Gebremedhin, and Md~Mostofa~Ali
  Patwary.
\newblock Fast maximum clique algorithms for large graphs.
\newblock In {\em Proceedings of the 23rd International Conference on World
  Wide Web}, pages 365--366, 2014.

\bibitem{roughgarden2021beyond}
Tim Roughgarden.
\newblock {\em Beyond the worst-case analysis of algorithms}.
\newblock Cambridge University Press, 2021.

\bibitem{von2015generating}
Moritz von Looz, Henning Meyerhenke, and Roman Prutkin.
\newblock Generating random hyperbolic graphs in subquadratic time.
\newblock In {\em Algorithms and Computation: 26th International Symposium,
  ISAAC 2015, Nagoya, Japan, December 9-11, 2015, Proceedings}, pages 467--478.
  Springer, 2015.

\end{thebibliography}

\end{document}